\documentclass[review,preprint,11pt,doublespacing]{elsarticle}

\usepackage{amssymb}
\usepackage{amsmath,graphicx}
\usepackage[hypcap=false,figurename=Figure,tablename=Table]{caption}
\usepackage[hypcap=false,figurename=Figure,tablename=Table]{subcaption}
\usepackage[mathcal]{euscript}
\usepackage{amsthm,enumerate,amsfonts,algorithmic}
\usepackage[boxruled,vlined]{algorithm2e}
\usepackage{times, url}
\usepackage{textcomp}
\usepackage{pifont, color}
\usepackage[utf8]{inputenc}
\usepackage{array}
\usepackage{sistyle, enumitem, hhline}
\usepackage{mathrsfs}
\usepackage{wasysym}
\usepackage{amsbsy} 
\usepackage{setspace}

\journal{Signal Processing}

\newtheorem{thm}{Theorem}
\newtheorem{lem}[thm]{Lemma}
\newdefinition{rmk}{Remark}
\newtheorem{prop}{Proposition}

\newtheorem{definition}{Definition}

\newcommand{\PL}{\mathrm{PL}}

\newcommand{\bs}{\boldsymbol}

\newcommand{\bR}{{\mathbb R}}

\newcommand{\bE}{{\mathbb E}}
\newcommand{\E}{{\mathbb E}}

\newcommand{\cF}{{\mathcal F}}
\newcommand{\cK}{{\mathcal K}}
\newcommand{\la}{\langle}
\newcommand{\ra}{\rangle}
\renewcommand{\thetable}{\arabic{table}}
\renewcommand{\thefigure}{\arabic{figure}}

\begin{document}

\begin{frontmatter}



\title{Distributed on-line multidimensional scaling \\ for self-localization in wireless sensor networks}

 \author[ad1]{G.~Morral\fnref{fn1}}
 \ead{morralad@telecom-paristech.fr}


 \author[ad1]{P. Bianchi}
 \ead{bianchi@telecom-paristech.fr}

 \fntext[fn1]{The work of G. Morral is supported by DGA (French Armement Procurement Agency) and the Institut Mines-Télécom.
This work has been supported by the ANR grant ODISSEE.}

 \address[ad1]{Institut Mines-Télécom, Télécom ParisTech, CNRS  LTCI.
46, rue Barrault, 75013 Paris, France.}


\begin{abstract}
  The present work considers the localization problem in wireless
  sensor networks formed by fixed nodes. Each node seeks to estimate
  its own position based on noisy measurements of the relative
  distance to other nodes. In a centralized batch mode, positions
  can be retrieved (up to a rigid transformation) by applying
  Principal Component Analysis (PCA) on a so-called similarity matrix
  built from the relative distances. In this paper, we propose a
  distributed on-line algorithm allowing each node to estimate its own
  position based on limited exchange of information in the network. Our
  framework encompasses the case of sporadic measurements and random
  link failures. We prove the consistency of our algorithm in the case
  of fixed sensors. Finally, we provide numerical and experimental
  results from both simulated and real data. Simulations issued to
  real data are conducted on a wireless sensor network testbed.
\end{abstract}

\begin{keyword}
Principal component analysis \sep Wireless sensor networks \sep Distributed stochastic approximation algorithms \sep Localization \sep Multidimensional scaling \sep Received signal strength indicator


\end{keyword}

\end{frontmatter}


\section{Introduction} 
The problem of self-localization involving low-cost radio devices in
WSN can be viewed as an example of the internet of things (IoT). The
evolution in the last $50$ years of the embedded systems and smart
grids has contributed to enable the WSN integrates the emerging system
of the IoT. Recently, advanced applications to handle specific tasks
require the support of networking features to design cloud-based
architectures involving sensor nodes, computers and other remote
component. Among the large range of applications, location services
can be provided by small devices carried by persons or deployed in a
given area, \emph{e.g.} routing and querying purposes, environmental
monitoring, home automation services.

In this paper we investigate the problem of localization in wireless
sensor networks (WSN) as a particular application of principal
component analysis (PCA).
We assume that wireless sensor devices are able to obtain
received signal strength indicator (RSSI) measurements that can be related to a ranging model
depending on the inter-sensor distances. The multidimensional scaling mapping method (MDS-MAP)
consists in applying PCA to a so-called similarity matrix constructed from
the squared inter-sensor distances. Then, the sensors' positions can be recovered 
(up to a rigid transformation) from the principal components of the similarity matrix~\cite{mds:borg},~\cite{shang:2003}. As opposed to time difference of arrival (TDOA) and angle of arrival
(AOA) techniques, the MDS-MAP approach allows to recover the network configuration
based on the sole RSSI, and can be used without any additional hardware or/and synchronization
specifically devote to self-localization.

MDS-MAP has been extensively studied in the literature (see Section~\ref{sec:state5} for an overview).
The algorithm is generally implemented in a centralized fashion. This requires the presence of a fusion
center which gather sensors' measurements, computes the similarity matrix, performs the PCA, and eventually
sends the positions to the respective sensors. In this paper, we provide a fully {\bf distributed}
algorithm which do not require RSSI measurements to be shared. In addition, our algorithm can be used {\bf on-line}.
By on-line, we mean that the current estimates of the sensors' positions are updated each time
new RSSI measurements are performed, as opposed to batch methods which assume that measurements are collected \emph{prior}
to the localization step. Therefore, although we assume throughout the paper that the sensors' 
positions are fixed, our algorithm has the potential to
be generalized to moving sensors, with aim to track positions while sensors are moving.

The paper is organized as follows. In Section~\ref{sec:fmw}, we provide the network and the observation models.
We also provide a brief overview of standard
self-localization techniques for WSN. Section \ref{sec:MDS} presents the centralized version of the MDS-MAP algorithm.
The proposed distributed MDS-MAP
algorithm is provided in Section~\ref{sec:dmds}.  An additional
refinement phase is also proposed in Section~\ref{sec:dmle} where our
MDS-MAP algorithm is coupled with a distributed maximum-likelihood
estimator. In Section~\ref{sec:sim}, numerical experiments based on
both simulated and real data are provided. Section~\ref{sec:conclusion} gives
some concluding remarks.

\section{The framework}
\label{sec:fmw}
\subsection{Network model}

Consider $N$ agents (\emph{e.g.} sensor nodes or other electronic
devices) seeking to estimate their respective positions defined as
$\{\bs z_1,\cdots,\bs z_N\}$ where for any $i$, ${\bs z}_i\in \bR^p$ with
$p=2$ or 3. 
We assume that agents have only access to noisy measurements of their relative RSSI values. 
More precisely, each agent $i$ observes some RSSI measurements $P_{i,j}$ associated with
other agents $j\neq i$. Here, $P_{i,j}$ is a random function of the Euclidean distance
$d_{i,j}= \|\bs z_i - \bs z_j\|$ between nodes $i$ and $j$. 
The statistical model relating RSSI values to inter-sensor distances is provided in the next paragraph.

The goal is to design a distributed and
on-line algorithm to enable each sensor node to estimate its position
$\bs z_i$ from noisy measurements of the distances. 
Before going further in the description of the RSSI statistical model,
it is worth noting that the localization problem is in fact
ill-posed. Since the only input data are distances, exact positions
are identifiable only up to a rigid transformation. Indeed, quantities
$(d_{i,j})_{\forall i,j}$ are preserved when an isometry is applied to
the agents' positions, \emph{i.e.} rotation and translation. The
problem is generally circumvented by assuming a minimum number of
\emph{anchors} or also named \emph{landmarks} (sensor nodes whose
GPS-positions are known), \emph{e.g.} $M=3$ or $4$ when $p=2$, and
considering these prior knowledge to identify the indeterminacy.
This point is further discussed in Section~\ref{sec:state5}.

\subsection{Received signal model}
\label{sec:problem5}

We rely on the so-called log-normal shadowing model (LNSM) to model
RSSI measurements as a function of the inter-sensor distance \cite{rappaport:2002}.
We define the average path loss $\PL(d)$ at a distance
$d$ expressed in $\SI{}{dB}$ as $\PL(d) = \PL_0 + 10\eta \log_{10}
\frac{d}{d_0}$, where the parameters $\eta$, $d_0$ and $\PL_0$ depend
on the environment (see Section~\ref{sec:sim}).
Given that the distance between sensors $i$ and $j$ is $d_{i,j}$, we define the
RSSI between $i$ and $j$ as a random variable $P_{i,j}$ satisfying
\begin{align}
\label{eq:rssi}
P_{i,j}  = -\PL(d_{i,j}) + \epsilon_{i,j}
\end{align}
where $(\epsilon_{i,j}:i\neq j)$ are thermal noises assumed independent with zero mean and variance $\sigma^2$.
Assume that a given agent $i$ is provided with $T$ independent copies $P_{i,j}(1),\dots, P_{i,j}(T)$ of the random variable $P_{i,j}$
and let $\bar P_{i,j} = T^{-1}\sum_{t=1}^TP_{i,j}(t)$ be the empirical average. An unbiased estimate of the
 squared distance $d_{i,j}^2$ is given by
\begin{equation}
\label{eq:d2nobias}
{D}(i,j) = \frac{10}{C^4}^{\frac{-\bar{P}_{i,j}-\PL_0}{5\eta}} 
\end{equation}
where  $C=10^{\frac{\sigma^2 \ln10}{2T(10\eta)^2}}$. 
Indeed, it can be easily checked that the mean and variance of the unbiased estimator~\eqref{eq:d2nobias} are respectively: $\bE[{D}(i,j)] = d_{i,j}^2$ and $\bE[({D}(i,j)-d_{i,j}^2)^2] = d_{i,j}^4(C^8-1)$.
The construction of unbiased estimates of squared distance will be the basic ingredient of our distributed MDS-MAP algorithm.

\subsection{Overview of some localization techniques}
\label{sec:state5}

Several overview papers have been published in the last ten years dealing with the classification of the localization techniques (see \cite{patwari:2005} or \cite{mao:2007}). In some situations, localization is made easier by the presence of \emph{anchor}-nodes whose
positions are assumed perfectly known. Other methods, called anchor-free, do not require the presence of such landmarks.

%

\textbf{Anchor-based methods: }The classical techniques involve the
resolution of a single unknown position of a sensor node at a time by
means of RSSI values following the LNSM coming from a fixed number of
surrounding anchor nodes or landmarks. Since the sensor node only uses
the information from known positions, its position can be expressed in
absolute coordinates, \emph{i.e.} anchor positions in
GPS-coordinates. When considering a noisy scenario, several works
coupled the classical methods (\emph{trilateration},
\emph{multilateration}~\cite{multi:2011} or\emph{min-max}
\cite{savvides:2002}) with a least squares problem. In particular,
\cite{niculescu:2001}, \cite{savarese:2002} and
\cite{savvides:2002} consider multi-hop communications between the
sensor nodes. Other approaches focus on the statistical distribution
of the received RSSI measurements coming from the landmarks. The goal
is to consider a parametric model for the received signal and to apply
maximum likelihood estimator (MLE). Most works consider the
LNSM (see for instance \cite{patwari:2001} or
\cite{kamol:2004}) while others assume alternative
statistical models (see \cite{mle:2010} or \cite{amy:2013}).

\textbf{Anchor-free methods: }The configuration of the network can be
recovered on a relative coordinate system instead of the GPS absolute
coordinate system. When distances between nodes are view as similarity
metrics, the positioning problem is referred to multidimensional
scaling (MDS). The aim is
to find an embedding from the $N$ nodes such that distances are
preserved. In classical MDS~\cite[Chapter 12]{mds:borg}
positions are obtained by principal component analysis (PCA) of a
 $N\times N$ matrix constructed from the Euclidean distances. If
distances are issued to some noise, \emph{e.g.} estimated from RSSI
measurements as~\eqref{eq:d2nobias}, \cite{shang:2003} propose a
MDS-MAP algorithm based on the classical MDS problem. Indeed, the WSN
localization problem is solved by enabling each sensor node to infer all
the estimated pairwise distances. Alternative approaches within the
localization context are based on optimization techniques. In metric
MDS, positions are obtained by the stress majorization
algorithm  SMACOF (see \cite[Chapter 8]{mds:borg} and
\cite{leeuw:1977}). Alternatively, semidefinite programming (SDP)
can be used as in \cite{biswas:2006:cent}.

The latter approaches have been also addressed in a distributed
setting without the presence of a central processing unit. A distributed
batch version of the SMACOF algorithm based on a round-robin
communication scheme is proposed in \cite{costa:2005}. Since
\cite{costa:2005} considers the minimization of the non-convex stress
function, the same distributed approach (batch and incremental) is
presented in \cite{cedric:2007} but using a quadratic criterion which
includes the information from the anchor nodes to overcome the
non-convex issue. The Authors of \cite{biswas:2006:cent} propose a
distributed implementation of their SDP-based localization
algorithm. In~\cite{biswas:2006:dist} the network is divided in
several clusters of at least two anchor nodes and a large number of
sensor nodes and then the SDP problem is addressed locally at each
cluster. More recently, gossip-based algorithms have been proposed in
\cite{sayed:2010}, \cite{gianackis:2009} to solve the distributed
optimization problem via Kalman filtering and gradient descent
approaches. Other works address the distributed WSN localization
problem using the multidimensional scaling (MDS) method based on
PCA. The MDS-MAP proposed in \cite{shang:2003} is later improved in
\cite{shang:2004}. In~\cite{shang:2004} each sensor node applies the
MDS-MAP of \cite{shang:2003} to its local map and then the local maps
are merged sequentially to recover the global map. Alternatively,
in~\cite{montanari:2011} and~\cite{cdc:2012} a sparsification matrix model on the observations is
introduced to decentralized the PCA step.

\section{Centralized MDS-MAP}
 \label{sec:MDS}

\subsection{Centralized batch MDS}
\label{sec:batchmds}

Define $\bs S$ as the $N\times N$ matrix of square relative distances \emph{i.e.}, $\bs S(i,j) = d_{i,j}^2$. Define $\overline {\bs z}=\frac 1N\sum_{i=1}^N\bs z_i$ as the center of mass (or~\emph{barycenter}) of the agents. Upon noting that $d_{i,j}^2 = \|\bs z_i-\overline{\bs z}\|^2+\|\bs z_j-\overline {\bs z}\|^2 - 2\la \bs z_i-\overline {\bs z},\bs z_j-\overline {\bs z}\ra$, one has:
\begin{equation}
\bs S = \bs c\bs 1^T + \bs 1\bs c^T-2 \bs Z\bs Z^T
\label{eq:S}
\end{equation}
where $\bs 1$ is the $N\times p$ matrix whose components are all equal to one, $\bs c = (\|\bs z_1-\overline {\bs z}\|^2,\cdots, \|\bs z_N-\overline{\bs z}\|^2)^T$ and the $i$th line of matrix $\bs Z$ coincides with the row-vector $\bs z_i-\overline {\bs z}$.
Otherwise stated, the $i$th line of $\bs Z$ coincides with the \emph{barycentric coordinates} of node $i$.
Define $\bs J=\bs 1\bs 1^T/N$ as the orthogonal projector onto the linear span of the vector $\bs 1 = (1,\dots,1)^T$.
Define $\bs J_\bot = \bs I_N-\bs J$ as the projector onto the space of vectors with zero sum, where $\bs I_N$ is the $N\times N$ identity matrix. 
It is straightforward to verify that $\bs J_\bot \bs Z = \bs Z$. Thus, introducing the matrix
\begin{align}
\label{eq:M}
\bs M\triangleq -\frac 12\bs J_\bot \bs S \bs J_\bot\,,
\end{align}
equation~(\ref{eq:S}) implies that $\bs{M}=\bs Z \bs Z^T$. In
particular, $\bs{M}$ is symmetric, non-negative and has rank (at most)
$p$. The agents' coordinates can be recovered from $\bs{M}$ (up to a
rigid transformation) by recovering the principal eigenspace of
$\bs{M}$ \emph{i.e.}, the vector-space spanned by the $p$th principal
eigenvectors (see~\cite[Chapter 12]{mds:borg}). 


 Denote by $\{\lambda_k\}_{k=1}^N$ the eigenvalues of $\bs{M}$ in
 decreasing order, \emph{i.e.} $\lambda_1\geq \cdots \geq
 \lambda_N$. In the sequel, we shall always assume that $\lambda_p>0$. 
Denote by $\{\bs
 u_k\}_{k=1}^p$ corresponding unit-norm $N\times 1$ eigenvectors. Set
 $\bs Z = (\sqrt{\lambda_1} \bs u_1,\cdots,\sqrt{\lambda_p}\bs
 u_p)$. Clearly $\bs{M}=\bs Z\bs Z^T = \bar {\bs Z}\bar {\bs Z}$ and $\bar{\bs Z} = \bs R\bs Z$ for
 some matrix $\bs R$ such that $\bs R\bs R^T=\bs I_N$.  Otherwise
 stated, $\bar{\bs Z}$ coincides with the barycentric coordinates $\bs Z$ up
 to an orthogonal transformation. In particular, the $i$th row of matrix $\bar{\bs Z}$
is an estimate of the position of the $i$th sensor (up to the latter transformation common to all sensors).
In practice, matrix $\bs S$ is
 usually not perfectly known and must be replaced by an estimate
 $\bs{\widehat S}$. This yields the Algorithm~\ref{alg:mds}
 (see~\cite[Chapter 12]{mds:borg}).

\footnotesize

\renewcommand{\algorithmicrequire}{\textbf{Initialize:}}
\newcommand{\INITIALIZE}{\REQUIRE}
\renewcommand{\algorithmicensure}{\textbf{Iterate:}}
\newcommand{\UPDATE}{\ENSURE}
  \begin{algorithm}
\caption{Centralized batch MDS-MAP for localization}
  \label{alg:mds}

{\bf Input}: Noisy estimates of the square distances $D(i,j)$~\eqref{eq:d2nobias} for all pair $i,j$.

1. Compute matrix $\bs{\widehat S}= (D(i,j))_{i,j=1,\dots,N}$.

2. Set $\bs{\widehat{M}} = -\frac 12\bs J_\bot \bs{\widehat S}\bs  J_\bot$.

3. Find the eigenvectors $\{\bs u_k\}_{k=1}^p$ and eigenvalues $\{\lambda_k\}_{k=1}^p$ of $\bs{\widehat{M}}$.

{\bf Output}: $\bs{\widehat Z}=(\sqrt{\lambda_1} \bs u_1,\cdots,\sqrt{\lambda_p}\bs u_p)$
\end{algorithm}

\normalsize


 \subsection{Centralized on-line MDS}
\label{sec:cmds}

In the previous batch Algorithm~\ref{alg:mds}, measurements are made prior to the estimation of the coordinates. From now on, observations are not stored into the system's memory: they are deleted after use. Thus, agents gather measurements of their relative distance with other agents and, simultaneously, estimate their position. 

\subsubsection{Observation model: sparse measurements}
\label{sec:obs}

We introduce a collection of independent r.v.'s $(P_{i,j}(n) :
i,j=1,\cdots,N,\,n\in \mathbb N)$ such that each $P_{i,j}(n)$ follows
the LNSM described in Section~\ref{sec:problem5}. At time $n$, it is possible to define
an unbiased estimate $\bs
D_n(i,j)$ the squared distance as $\bs D_n(i,j)=
\frac{10}{C^4}^{\frac{-P_{i,j}(n)-\PL_0}{5\eta}}$ in the sense that
$\bE[\bs D_n(i,j)] = d_{i,j}^2$. We use the convention that $\bs D_n(i,i)=0$.

\begin{definition}[Sparse measurements]
\label{def:obs}
At each time instant $n$, we assume that with probability $q_{ij}$, an
agent $i$ is able to obtain an estimate $\bs S_n(i, j)$ of the square
distance with an other agent $j \neq i$ and makes no observation
otherwise. Thus, one can represent the available observations as the
product $\bs S_n(i,j) = \bs A_n(i,j)\bs D_n(i,j)$ where $(\bs A_n)_n$
is an i.i.d. sequence of random matrices whose components $\bs
A_n(i,j)$ follow the Bernoulli distribution of parameter
$q_{ij}$. Stated otherwise, node $i$ observes the $i$th row of matrix
$\bs A_n\circ \bs D_n$ at time $n$ where $\circ$ stands for the
Hadamard product.
\end{definition}

\begin{lem}
\label{lem:sn}
  Assume $q_{ij}>0$ for all pairs $i,j$. Set $\bs W := [q_{ij}^{-1}]_{i,j=1}^N$ and let $\bs A_n$, $\bs S_n$ be defined as above.
The matrix
\begin{equation}
\label{eq:Sn}
\bs S_n = \bs W\circ \bs A_n\circ\bs  D_n
\end{equation}
is an unbiased estimate of $\bs S$ \emph{i.e.}, $\mathbb E[\bs S_n] = \bs S$.
\end{lem}
\begin{proof}
Each entry of matrix $\bs S_n$, $\bs S_n(i,j)$, is equal to $1/q_{ij}\bs A_n(i,j)\bs D_n(i,j)$. As the random variables $\bs A_n(i,j)$ and $\bs D_n(i,j)$ are independent, by the above definition of $\bs D_n$ and $\mathbb E[\bs A_n(i,j)] = q_{ij}$, then $\mathbb E[\bs S_n(i,j)] = d_{i,j}^2$.
\end{proof}
As a consequence of Lemma~\ref{lem:sn}, an unbiased estimate of $\bs M$ defined in \eqref{eq:M} is simply obtained by
$\bs{M}_n = -\frac 12 \bs J_\bot \bs S_n\bs J_\bot$.

\subsubsection{Oja's algorithm for the localization problem}
\label{sec:doja}

When dealing
with random matrices $\bs M_n$ having a given expectation $\bs{M}$,
the principal eigenspace of $\bs M$ can be recovered by 
the Oja's algorithm~\cite{oja:1992}. The latter consists in recursively defining
a sequence $\bs U_n$ of $N\times p$ matrices, which stand for the estimate at time $n$ of the $p$ principal unit-eigenvectors
of $\bs M$. The iterations as firstly introduced in~\cite{oja:1992} are given by:
\begin{equation}
\label{eq:oja}
\bs U_{n} = \bs U_{n-1} + \gamma_{n}\left( \bs M_{n}\bs U_{n-1} - \bs U_n\big(\bs U_{n-1}^T \bs M_{n}\bs U_{n-1}\big)\right)\ ,
\end{equation}
where $\gamma_n>0$ is a step size. Note that in practice, the algorithm is likely to suffer from numerical instabilities.
In~\cite{borkar:meyn:2012}, a renormalization step is introduced to avoid unstabilities. 
As this approach seems difficult to generalize in a distributed context, it is more adequate in our context to 
introduce a reprojection step in~(\ref{eq:oja}) of the form
$$
\bs U_{n} = \Pi_{\mathcal K}\left[\bs U_{n-1} + \gamma_{n}\left( \bs M_{n}\bs U_{n-1} - \bs U_n\big(\bs U_{n-1}^T \bs M_{n}\bs U_{n-1}\big)\right)\right]\ ,
$$
where $\Pi_\cK$ is a projector onto an arbitrarily large convex compact set $\cK$ chosen large enough to include all matrices whose columns have unit-norm.
Typically, we set ${\mathcal K} = [-\alpha,\alpha]^p \times \dots \times
[-\alpha,\alpha]^p$ where  $\alpha>1$.

In order to obtain an estimate of the sensors positions, we also need to estimate the principal eigenvalues
in addition to the eigenvectors.
Let $\bs
u_{n,k}$ denote the $k$th column of matrix $\bs U_n$. 
Define the quantity $\lambda_{n,k}$ recursively by:
\begin{align}
\label{eq:oja2}
 \lambda_{n,k} =\lambda_{n-1,k} +\gamma_n\left(\bs u_{n-1,k}^T\bs M_{n}\bs u_{n-1,k}-\lambda_{n-1,k} \right)\,.
\end{align}
The convergence properties of Oja's algorithm are studied in details in~\cite{oja:1992} and \cite{borkar:meyn:2012}.
Finally, according to step 3 of the batch Algorithm~\ref{alg:mds}, the estimated barycentric coordinates are obtained as:
\begin{align}
\label{eq:oja3}
\bs{\widehat Z}_n = \left(\sqrt{\lambda_{n,1}}\bs u_{n,1},\dots,\sqrt{\lambda_{n,p}}\bs u_{n,p}\right).
\end{align}
The combination of Equations~(\ref{eq:oja})~(\ref{eq:oja2}) and~(\ref{eq:oja3}) provides an on-line 
for MDS-MAP algorithm. However, the computation of matrix ${\bs M}_n$ at each step as well as the matrix products in~(\ref{eq:oja})
require a full amount of centralization.


\section{Distributed on-line MDS-MAP}
\label{sec:dmds}


\subsection{Communication model}

It is clear from the previous section that an unbiased estimate of matrix $\bs M$ is the first step needed to estimate the sought eigenspace. In the centralized setting, this estimate was given by matrix $\bs M_n= -\frac 12\bs J_\bot \bs S_n \bs J_\bot$.
As made clear by the observation model (in Definition~\ref{def:obs}), each node $i$ observes the~$i$th row of matrix~$\bs S_n$. As a consequence, node $i$ has access to the $i$th row-average $\bs{\overline S}_n(i)\triangleq \frac 1N\sum_j\bs S_n(i,j)$. This means that matrix $\bs S_n \bs J_\bot$ can be obtained with no need to further exchange of information in the network. On the other hand, $\bs J_\bot \bs S_n \bs J_\bot$ requires to compute the per-column averages of matrix $\bs S_n\bs J_\bot$, \emph{i.e.} $\frac 1N\sum_j\bs S_n(j,i)$ for all $i$. This task is difficult in a distributed setting, as it would require that all nodes share all their observations at any time.
A similar obstacle happens in Oja's algorithm when computing matrix products, \emph{e.g.} $\bs M_{n}\bs U_{n-1}$ in~\eqref{eq:oja}. To circumvent the above difficulties, we introduce the following sparse asynchronous communication framework.
 In order to derive an unbiased estimate of $\bs M$, let us first remark that for all $i$, $j$,
 \begin{equation}
 \label{eq:Mij}
 \bs M(i,j) =  \frac {\overline {d^2}(i)+\overline {d^2}(j)}{2} -\frac {d_{i,j}^2+\delta }2
 \end{equation}
 where we set $\overline {d^2}(i)\triangleq  \frac 1{N}\sum_{k} d_{ik}^2$ and $\delta\triangleq \frac 1N\sum_i\overline {d^2}(i)$.
 Note that the terms $d_{i,j}^2$ and $\overline {d^2}(i)$ can be estimated by $\bs S_n(i,j)$ and $\bs{\overline S}_n(i)$ respectively.
However, additional communication is needed to estimate $\delta$ since it corresponds to the average value over all square distances.
We define
 \begin{equation}
 \label{eq:hatM}
 \bs{\widehat M}_n(i,j) = \frac { \bs{\overline {S}}_n(i) + \bs{\overline {S}}_n(j)} {2} -\frac {\bs S_n(i,j)+\bs \delta_n(i) }2 
 \end{equation}
where $\bs \delta_n(i)$ is a quantity that we will define in the sequel, and which represents the estimate of $\delta$ at the agent $n$.

We are now faced with two problems. First, we must construct $\bs \delta_n(i)$ as an unbiased estimate of $\delta$. 
Second, we need to avoid the computation of $\bs{\widehat M}_n(i,j)$ for \emph{all} pairs $i,j$, but only to some of them.
In order to provide an answer to these problems, we introduce the notion of asynchronous transmission sequence. Formally,
 \begin{definition}[Asynchronous Transmission Sequence]
  Let $q$ be a real number such that $0<q<1$.  We say that the sequence of random vectors
$T_n = (\iota_n,Q_{n,i}\,:\,i\in\{1,\cdots,N\}, n\in \mathbb N)$ is an Asynchronous Transmission Sequence (ATS) if:
\emph{i)} all variables $(\iota_n,Q_{n,i})_{i,n}$ are independent, \emph{ii)} $\iota_n$ is uniformly distributed on the set $\{1,\cdots,N\}$, \emph{iii)}
   $\forall i\neq \iota_n$, $Q_{n,i}$ is a Bernoulli variable with parameter $q$ \emph{i.e.}, $\mathbb P[Q_{n,i}=1]=q$ and
\emph{iv)} $Q_{n,\iota_n}=0$. 
\label{def:ats}
\end{definition}


Let $(T_n)_n$ denote an ATS defined as above.  At time $n$, we assume that a given node
$\iota_n\in \{1,\dots, N \}$ wakes up and transmits its
local row-average $\bs{ \overline S}_n(\iota_n)$ to other nodes.
All nodes $i$ such that $Q_{n,i}=1$ are supposed to receive the message. For any $i$, we set:
\begin{equation}
 \bs \delta_n(i) = \frac{\bs{\overline S}_n(i)}N + \frac{\bs{\overline S}_n(\iota_n)Q_{n,i}}q\,.\label{eq:ghjkl}
\end{equation}
 
 The following Lemma is a consequence of Definition~\ref{def:ats} along with Lemma~\ref{lem:sn} and equation~(\ref{eq:M}).
 \begin{lem}
   \label{lem:Munbiased}
 Assume that $(T_n)_n$ is an ATS independent of $(\bs S_n)_n$. Let $(\bs{\widehat M}_n)_n$ be the sequence of matrices defined
 by~(\ref{eq:hatM}). 
Then, $\mathbb E[\bs{\widehat M}_n] = M$.
 \end{lem}
 \begin{proof}
 By Lemma~\ref{lem:sn} the expectation of terms $\bs{\overline {S}}_n(i)$, $\bs{\overline {S}}_n(j)$ and $\bs S_n(i,j)$ are respectively $ \overline {d^2}(i)$, $\overline{d^2}(j)$ and $d_{i,j}^2$. Moreover, by Definition~\ref{def:ats} the expectation of the random term $\bs \delta_n(i)$ is equal to 
 $$
 \bE[\bs \delta_n(i) ] = \frac{1}{N}\bE[\bs{\overline{S}}_n(i)]+\frac{1}{q}\frac{1}{N}\sum_{j\neq i} \bE[\bs{\overline{S}}_n(j)]q=\frac{1}{N}\sum_{i=1}^N\overline {d^2}(i) \, ,
 $$ 
 which coincides with $\delta$. Then, the expectation of each entry of the matrix $\bs{\widehat M}_n$ in~\eqref{eq:hatM} is equal to the corresponding $\bs M(i,j)$ defined in~\eqref{eq:Mij}.
 \end{proof}

 \subsection{Preliminaries: constructing unbiased estimates}

 As we now obtain a distributed and unbiased estimate of $\bs M$, the
 remaining task is to adapt accordingly the Oja's
 algorithm~(\ref{eq:oja}). In this paragraph, we provide the main ideas behind the construction of
our algorithm.

Assume that we are given a current estimate $\bs U_{n-1}$ at time $n$, under the form of a $N\times p$ matrix.
Assume also that for each $i$, the $i$th row of $\bs U_{n-1}$ is a variable which is physically handled by node $i$.
We denote by $\bs U_{n-1}(i)$ the $i$th row of $\bs U_{n-1}$.

Looking at~(\ref{eq:oja}) in more details, 
Oja's algorithm requires the evaluation of intermediate values, 
as unbiased estimates of $\bs M \bs U_{n-1}$ and $\bs U_{n-1}^T\bs M \bs U_{n-1}$.

We consider the previous ATS $(T_n)_n$ involved in~\eqref{eq:hatM}. We assume that the active node $\iota_n$ (\emph{i.e.}, the one which transmits $\bs{\overline S}_n(\iota_n)$) is also able to transmit its local estimate $\bs U_{n-1}(\iota_n)$ at same time. Thus, with probability $\frac{1}{N}$, node $\iota_n$ sends its former estimate
$\bs U_{n-1}(\iota_n)$ and $\bs{\overline S}_n(\iota_n)$ to all nodes~$i$ such that $Q_{n,i}=1$. Then, all nodes compute:
\begin{equation}
\label{eq:Yn}
\bs Y_n(i) = \bs{\widehat M}_n(i,i)\bs U_{n-1}(i) + \frac Nq \bs U_{n-1}(\iota_n)\bs{\widehat M}_n(i,\iota_n)Q_{n,i}
\end{equation}
As it will be made clear below, the $N\times p$ matrix $\bs Y_n$ whose $i$th row coincides with $\bs Y_n(i)$ can be interpreted as an unbiased estimate of $\bs M\bs U_{n-1}$.

Now we introduce the distributed version of the second term $\bs U_{n-1}^T\bs M_n \bs U_{n-1}$. Consider a second ATS $(T_n')_n$ independent of $(T_n)_n$. 
At time $n$, node $\iota_n'$ wakes up uniformly random and broadcasts the product $\bs U_{n-1}(\iota_n')^T\bs Y_n(\iota_n')$ to other nodes. Receiving nodes are those $i$'s for which $Q_{n,i}'=1$. Then, all nodes are able to compute the estimate $p \times p$ matrix as follows: 
\begin{align}
  \label{eq:Sigman}
    \bs \Lambda_n(i) =  \bs U_{n-1}(i)^T\bs Y_n(i) + \frac Nq \bs U_{n-1}(\iota_n')^T\bs Y_n(\iota_n')Q_{n,i}' \,.
\end{align}

\begin{lem}
\label{lem:YnSigman}
Let $(T_n)_n$ and $(T_n')_n$ be two independent ATS. For any $n$, denote by ${\mathcal F}_n$ the $\sigma$-field generated by
$(T_k)_{k\leq n}$, $(T_k')_{k\leq n}$, $(A_k)_{k\leq n}$ and $(D_k)_{k\leq n}$. Let $(\bs U_n)_n$ be a ${\mathcal F}_n$-measurable $N\times p$ random matrix and
let $\bs Y_n$, $\bs \Lambda_n$ be defined as above. Then,
\begin{align*}
  {\mathbb E}[\bs Y_n|{\mathcal F}_{n-1}] = \bs M \bs U_{n-1} \quad \text{and} \quad \mathbb E[\bs \Lambda_n(i)|{\mathcal F}_{n-1}]  =\bs U_{n-1}^T\bs M\bs U_{n-1}\,.
\end{align*}
Under Lemma~\ref{lem:sn}, \ref{lem:Munbiased} and Definition~\ref{def:ats}, the random sequences $\bs Y_n(i)$ and $\bs \Lambda_n(i)$ are unbiased estimates of $\sum_j\bs M(i,j) \bs U_{n-1}(j)$ and $\bs U_{n-1}^T\bs M \bs U_{n-1}$ respectively given $\bs U_{n-1}$. 
\end{lem}

\begin{proof}
For each $i$, we obtain
\begin{align*}
\mathbb E[\bs Y_n(i)|{\mathcal F}_{n-1}] &= \bs M(i,i)\bs U_{n-1}(i) + \frac{N}{q}\frac{q}{N}\sum_{j\neq i}\bs M(i,j)\bs U_{n-1}(j) \\
&=\sum_j\bs M(i,j) \bs U_{n-1}(j) \, ,
\end{align*}
and
\begin{align*}
\mathbb E[\bs \Lambda_n(i)|{\mathcal F}_{n-1}] &= \bs U_{n-1}(i)^T\E[\bs Y_n(i)|{\mathcal F}_{n-1}] + \frac{N}{q}\frac{1}{N}\sum_{j\neq i}\bs U_{n-1}(j)^T
\bE[\bs Y_n(j)|{\mathcal F}_{n-1}] q \\
&= \sum_i \sum_j \bs U_{n-1}(i)^T\bs M(i,j)\bs U_{n-1}(j)
\end{align*}
which corresponds with the square matrix $\bs U_{n-1}^T\bs M\bs U_{n-1}$.
\end{proof}

\subsubsection{Main algorithm}

We are now ready to state the main algorithm.
The algorithm generates iteratively and for any node $i$ two variables $\bs U_{n}(i)$ and $\bs \lambda_n(i)$, according to:
\begin{align}
  \label{eq:updateU}
 & \bs U_{n}(i) = \bs U_{n-1}(i) + \gamma_{n}\left( \bs Y_n(i) - \bs U_{n-1}(i)\bs \Lambda_n(i)\right)\, \\
 \label{eq:vaps}
  & \bs \lambda_n(i) = \bs \lambda_{n-1}(i) + \gamma_n( \text{diag}(\bs \Lambda_n(i)) - \bs \lambda_{n-1}(i))\, .
\end{align}
For the same reasons as before, it is important in practice to introduce a projection step $\Pi_\cK$ in~(\ref{eq:updateU})
to avoid numerical unstabilities.
Finally, as in \eqref{eq:oja3}, each sensor $i$ obtains its estimate position $\bs{\widehat Z}_n(i)$ by: 
\begin{align}
\label{eq:pose}
\bs{\widehat Z}_n(i) = \left(\sqrt{\bs \lambda_{n,1}(i)}\bs u_{n,1}(i),\cdots,\sqrt{\bs \lambda_{n,p}(i)}\bs u_{n,p}(i)\right)
\end{align}
where we set $\bs U_n(i) = ( u_{n,1}(i), \dots, \bs u_{n,p}(i))$.

The proposed algorithm~\eqref{eq:updateU}-\eqref{eq:pose} is summarized in Algorithm~\ref{alg:dmds} below. 
Note that, at each iteration time $n$, only two communications are performed by two randomly selected nodes issued to the ATS's $T_n$ and $T_n'$.

\renewcommand{\algorithmicensure}{\textbf{Update:}}
\begin{algorithm}[h!]
  \caption{Distributed on-line MDS-MAP for localization (doMDS) }
\label{alg:dmds}
\begin{algorithmic}
\UPDATE  At each time $n=1,2,\dots$ 

\textbf{[Measures]:} each sensor node $i$, do:\\

 Makes sparse measurements of their RSSI to obtain $(\bs D_n(i,j))_j$ for some $j$ \\ such that $\, \bs A_n(i,j)=1$ (Definition~\ref{def:obs}). Set
\[\bs S_n(i,j) = \left\{
  \begin{array}{l l}
    q_{ij}^{-1} \bs D_n(i,j) & \text{ if }  \, \, \bs A_n(i,j)=1\\
    0 & \text{ otherwise }
  \end{array}\right.\]
and set $\, \, \bs{ \overline S}_n(i) = \frac 1N \sum_j \bs S_n(i,j)$.

\smallskip

\textbf{[Communication step]:} 

A randomly selected node  $\iota_n$ wakes up, then

$\quad$ \emph{i)} The node $\iota_n$ randomly selected broadcasts
  $\bs U_{n-1}(\iota_n)$ and $\bs{\overline S}_n(\iota_n)$ to \\ nodes $i$ such
  that $Q_{n,i}=1$.

$\quad$ \emph{ii)}   Each node $i$ computes $\bs Y_n(i)$ by (\ref{eq:Yn}).

$\quad$ \emph{iii)}  A node $\iota_n'$ randomly selected broadcasts $\bs U_{n-1}(\iota_n')^T\bs Y_{n}(\iota_n')$ to\\
  nodes $i$ such that $Q_{n,i}'=1$.

$\quad$ \emph{iv)}  Each node $i$ updates $\bs U_n(i)$ by~\eqref{eq:Sigman}-(\ref{eq:updateU}) and $\bs{\widehat Z}_n(i)$ by \eqref{eq:pose}.

\end{algorithmic}
\end{algorithm}

\normalsize

\subsection{Convergence analysis}
\label{sec:convergence}


We make the following assumptions.
The sequence $(\gamma_n)_n$ is positive and satisfies 
$$
\sum_n\gamma_n=+\infty\qquad\text{ and }\qquad\sum_n\gamma_n^2<\infty\,.
$$
Moreover we make the assumption that the sequence $\bs U_n$ remains a.s. in a fixed compact set $\cK$.
It must be emphasized that this assumption is difficult to check in practice. 
As mentioned above, stability can be enforced by means of a projection
step onto $\cK$.


\begin{prop}
\label{prop:meanfield}
For any $\bs U\in {\mathbb R}^{N\times p}$, set $h(\bs U) = \bs M\bs U-\bs U\bs U^T\bs M\bs U$. 
Let $\bs U_n$ be defined by (\ref{eq:updateU}).
There exists a random sequence $\xi_n$ such that, almost surely (a.s.),
$\sum_n\gamma_n\xi_n$ converges and
\begin{align}
\label{eq:algo}
\bs U_n = \bs U_{n-1} + \gamma_n h(\bs U_{n-1})+\gamma_n\xi_n\,.
\end{align}  
\end{prop}
The proof is provided in the Appendix. We are now in position to state the main convergence result.
\begin{thm}
\label{theo:convergence}
For any $k=1,\cdots,p$, the $k$th column $\bs u_{n,k}$ of $\bs U_n$ converges to an eigenvector of $\bs M$ with unit-norm.
Moreover, for each node $i$, $\bs \lambda_{n,k}(i)$ converges to the corresponding eigenvalue.
\end{thm}
The proof is provided in the Appendix.

Note that Theorem~\ref{theo:convergence} might seem incomplete in some
respect: one indeed expects that the sequence $\bs U_n$ converges to
the principal eigenspace of $\bs
M$. Instead, Theorem~\ref{theo:convergence} only guarantees that one
recovers \emph{some} eigenspace of $\bs M$. Undesired limit points can
be theoretically avoided by introducing an arbitrary small Gaussian
noise inside the parenthesis of the left hand side
of~(\ref{eq:updateU}) (see Chapter 4 in \cite{borkar:2008}). These so-called avoidance of traps techniques
are however out of the scope of this paper, and numerical results indicate that the principal
eigenspace is indeed recovered in practical situations.


\section{Position refinement: distributed maximum likelihood estimator}
\label{sec:dmle}

In the context of WSN localization, a refinement phase is in general
added (see \cite{savarese:2002},
\cite{savvides:2002}, \cite{shang:2004} or
\cite{costa:2005}). It is usually based on the statistical model relating
the observed RSSI values to the unknown positions, the latter being estimated
in the maximum likelihood sense.
The objective is twofolds. First, maximum likelihood estimation improves the
estimation accuracy obtained by the MDS-MAP approach. Second, as the MDS-MAP only identifies
positions up to a rigid transformation, it allows to eliminate the
residual ambiguity by using anchor nodes, provided that such anchors exist.

In this section, we provide a distributed algorithm in order to locally maximize the likelihood.
It is worth noting that the likelihood function is generally non-convex.
Thus, one cannot expect that a standard gradient ascent provides the maximum likelihood
estimator regardless from the initialization. For this reason, a preliminary phase
such as the proposed doMDS algorithm is essential as an initial coarse estimate, and the algorithm depicted below
is used merely as a fine search in the neighborhood of the doMDS output.


%

\subsection{Likelihood function}
\label{sec:dsa1}

Consider a connected graph $G=(V,E)$ where $V= \{1,\dots,N\}$ is the set of agents and
$E$ is a set of non-directed edges. 
In this paragraph, we allow for the presence of anchor nodes. We let $A\subset\{1,\dots, N\}$ be the set of anchor nodes
\emph{i.e.} for each $k\in A$, the position $\bs z_k$ of node $k$ is assumed to be known. 
Unknown parameters thus reduce to set of coordinates $\bs z=(\bs z_i:i\in \overline A)$ where $\overline A= V\backslash A$.
We denote by ${\mathscr N}_i$ the neighbors of $i$ which belong to $\overline A$ and by ${\mathscr M}_i$ the neighbors
of $i$ which are anchors. We note $\bs z_{\mathscr{N}_i}=(\bs z_j :
  j\in\mathscr{N}_i\cup \{i\})$.
For a connected pair of nodes $\{i,j\}$, 
we let $P_{i,j}(n)$ ($n\in \mathbb N$) be an i.i.d. sequence following the LNSM model of Section~\ref{sec:problem5}.
Equivalently, the quantity $\hat \ell_{i,j}(n) =\frac{-P_{i,j}(n)-\PL_0}{10\eta}$ follows a normal distribution with 
mean $\log_{10}d_{i,j}$ and variance $\frac{\sigma^2}{100\eta^2}$, since $\hat \ell_{i,j}(n)= \log_{10}d_{i,j}+\frac{\varepsilon_{i,j}}{10\eta}$ by using~\eqref{eq:rssi}.
Based on the observations $(\ell_{i,j}(n) : i\sim j)$ at a given time $n$, 
the likelihood associated with the unknown sensors' positions can be decomposed as
$$
{\mathcal L}_n(\bs z) =\sum_{i=1}^N f_{i,n}(\bs z_{\mathscr{N}_i}) 
$$
where
$$
f_i(\bs z_{\mathscr{N}_i}) = \sum_{j\in\mathscr{N}_i}\left( \hat \ell_{i,j}(n) - \log_{10}\|\bs z_i-\bs z_j\|\right)^2+ \sum_{k\in\mathscr{M}_i}\left( \hat \ell_{ik}(n) - \log_{10}\|\bs z_i-\bs z_k\|\right)^2\,.
$$

\subsection{The algorithm: on-line gossip-based implementation}
\label{sec:doMLE}

Following the idea of \cite{tsitsiklis:phd-1984}(see also \cite{dsa:2013} and reference therein),
we propose a distributed consensus-based implementation consisting on
local computations and random communications among the sensor
nodes. The algorithm is given below. The convergence proof is omitted due to the lack of space but follows from the same arguments as~\cite{dsa:2013}.

\renewcommand{\algorithmicrequire}{\textbf{Initialize:}}
\renewcommand{\algorithmicensure}{\textbf{Update:}}
\begin{algorithm}[h!]
  \caption{Distributed on-line MLE (doMLE)}
\label{alg:refin}
\begin{algorithmic}
\UPDATE at each time $n=1,2,\dots$ 

[\textbf{Local step}] each node $i$ obtains $\{P_{i,j}(n)\}_{\forall j\in\mathscr{N}_i} $ and $\{P_{ik}(n)\}_{\forall k\in\mathscr{M}_i} $.

Each sensor $i$ computes a temporary estimate of its position's set:
\vspace{-0.2cm}
$$
\tilde{\bs z}_{\mathscr{N}_i,n} = {\bs z}_{\mathscr{N}_i,n-1} - \gamma_n \nabla f_{i,n}({\bs z}_{\mathscr{N}_i,n-1}) 
$$
\vspace{-0.1cm}
[\textbf{Gossip step}] two uniformly random selected nodes $i\sim j$ in $\overline A$ exchange\\ their temporary estimated positions and average their values
according to:

\vspace{-0.6cm}

\begin{align*}
\!\!\!\!\!\!\!\!\!\!\!\!\!\!\!\!\!\!\!\!\!\!\!\!\!\!\!\!\!\!\!\!\!\!\!\!\forall \ell \in \mathscr N_i\cap \mathscr N_j,\quad {\bs z}_{\mathscr{N}_i,n}(\ell) &= \frac{\tilde{\bs z}_{\mathscr{N}_i,n}(\ell)+\tilde{\bs z}_{\mathscr{N}_j,n}(\ell)}2 \\ 
{\bs z}_{\mathscr{N}_j,n}(\ell) &= {\bs z}_{\mathscr{N}_i,n}(\ell), 
\end{align*}
Otherwise, $\forall \ell \notin \mathscr N_i\cap \mathscr N_j$ or $m \neq i,j$, set ${\bs z}_{\mathscr{N}_m,n}(\ell)=\tilde {\bs z}_{\mathscr{N}_m,n}(\ell)$.
\end{algorithmic}
\end{algorithm}

Algorithm~\ref{alg:refin} uses a standard pairwise averaging between nodes. We note that more involved gossip protocols have been proposed, we mention for instance broadcast and push-sum protocols (see~\cite{dsa:2014} and~\cite{nedic:olshevsky:2014}). Although theoretically possible, such an extension of Algorithm~\ref{alg:refin} is however beyond the scope of this paper.

\section{Numerical results}
\label{sec:sim}
We consider the same network configuration corresponding on the set of $N=50$ sensor nodes selected from the FIT IoT-LAB~\footnote{FIT IoT-LAB \url{https://www.iot-lab.info/}} platform at Rennes. Sensor nodes are located within a $\SI{5\times 9}{m^2}$ area, \emph{i.e.} $p=2$. Six sensors of the $50$ were set as anchor nodes (or landmarks). We compare the performance of our proposed distributed on-line MDS (doMDS) to other existing algorithms. We consider the distributed batch MDS~\cite{costa:2005} (dwMDS) and the classical centralized methods such as: multilateration~\cite{multi:2011} (MC), \emph{min-max}~\cite{savvides:2002}, Algorithm~\ref{alg:mds} in Section~\ref{sec:batchmds} (batch MDS) and the Oja's algoritm \eqref{eq:oja}-\eqref{eq:oja2} described in Section~\ref{sec:cmds}. The three iterative algorithms (Oja's, dwMDS and doMDS) are initialized by randomly chosen positions in $\SI{5\times 9}{m^2}$. 

\subsection{Simulated data}
\label{sec:sim_data}
 
First, we show the results from simulated data drawn according to the observation model defined in Section~\ref{sec:cmds}. In order to compare our proposed algorithm with the distributed MDS proposed by \cite{costa:2005}, we set the same environmental context in which $\sigma/\eta=1.7$. Figure~\ref{fig:comparison_veps} displays the comparison of the root-mean square error (RMSE) when running Algorithm~\ref{alg:dmds} over $300$ independent runs of the estimated positions when considering different communication parameters: $(q_{ij})_{i,j}$ (the Bernoullis related to the observation model~\eqref{eq:Sn}) and $q$ (the Bernoullis related to the ATS in Definition~\ref{def:ats}). Since the variance of the error sequence is upper bounded by the minimum probability value in~\eqref{eq:Exi-1}-~\eqref{eq:Exi-3}, we observe from Figure~\ref{fig:comparison_veps} a trade-off between the accuracy and the number of communications as the RMSE decreases faster when the probability $q$ is closer to $1$.  

 
Figure~\ref{subfig:comparison} shows the comparison of the localization RMSE over $300$ independent runs of the overall estimated positions when considering the three iterative methods: the centralized Oja's \eqref{eq:oja}-\eqref{eq:oja2}, the dwMDS of \cite{costa:2005} and our proposed Algorithm~\ref{alg:dmds}. The estimated positions after $1000$ iterations of the three iterative algorithms are reported in Figure~\ref{fig:sim_networks}. Note that, the result in Figure \ref{subfig:sim_dwdms} requires at least twice the number of communications compared to the results both on-line Oja's approaches. Positions close to the barycentric of the network tend to be more accurate than positions on the surrounding area for the three cases. Nevertheless, Figures~\ref{subfig:sim_comds} and \ref{subfig:sim_domds} show these outer positions better preserved than \cite{costa:2005}. Indeed, our distributed and asynchronous Oja's algorithm achieves in general better accuracy (around the $65$\% of positions) except for the third part of nodes which are located around the network's boundary, \emph{e.g.} nodes $11$ or $36-37$ for instance (see squared nodes in Figure~\ref{subfig:sim_domds}).

\subsection{Real data: FIT IoT-LAB platform of wireless sensor nodes}
\label{sec:real_data}

\subsubsection{Platform description}
In order to obtain real RSSI values we make use of the FIT IoT-LAB platform deployed at Rennes (France). The $256$ WSN430 open nodes\footnote{See the technical specifications of WSN430 sensors \url{https://github.com/iot-lab/iot-lab/wiki/Hardware_Wsn430-node} and CC2420 transceivers involve in our campaigns: \url{http://www.ti.com/lit/ds/symlink/cc2420.pdf}} available at the platform are issued to the standard ZigBee IEEE 802.15.4 operating at $\SI{2.4}{GHz}$. The sensor nodes are located in two storage rooms of size $\SI{6\times 15}{m^2}$ containing different objects. They are placed at the ceil which is $\SI{1.9}{m}$ height from the floor in a grid organization. Through of our user profile created in the FIT IoT-LAB's website, we run remotely several experiments involving the $50$ selected sensor nodes within $\SI{5\times 9}{m^2}$. All real data used in this section can be found in~\footnote{Data base available at G. Morral personal website \url{http://perso.telecom-paristech.fr/~morralad/}}. The environment parameters issued to the LNSM \eqref{eq:rssi} are: $\sigma^2=\SI{28.16}{dB}$, $\PL_0=\SI{-61.71}{dB} $ and $\eta=2.44$. We set $q_{ij}=0.8$ $\forall i,j$, $q=0.85$ and $\gamma_n =\frac{ 0.015}{\sqrt{n}}$ for Algorithm~\ref{alg:dmds}.

\subsubsection{Performance comparison}

We compare the same algorithms considered in Section~\ref{sec:sim_data} by setting the estimated positions obtained from each algorithm to the initialization of Algorithm~\ref{alg:refin}. Table~\ref{tab:comp_ref} shows the RMSE values before and after the refinement phase. In addition, we include the ratio of the accuracy improvement considering the RMSE values after and before applying the distributed MLE and the ratio regarding the number of positions over the total $N$ that are improved. The best performances are achieved by min-max, dwMDS and doMDS in terms of minimum RMSE value over the $N$ estimated positions. Nevertheless, the highest improvement is obtained with the proposed doMDS since the RMSE before the refinement phase was higher than the values from min-max and dwMDS which do not experiment a considerable decrease. In general, the refinement Algorithm~\ref{alg:dmds} improves almost all the positions for each method and especially the anchor-free methods based on the MDS approach. Indeed, the highest values are those from the distributed versions which may exploit in advantage the local knowledge of each sensor node. 

\section{Conclusion}
\label{sec:conclusion}

This paper introduced a novel algorithm based on Oja's algorithm for self-localization in wireless sensor networks. Our algorithm is based on a distributed PCA of a similarity matrix which is learned on-line. Almost sure convergence of the method is demonstrated in the context of vanishing step size. The algorithm can be coupled with a distributed maximum likelihood estimator to refine the sensors positions if needed.  Numerical results have been conducted on both simulated and real data  on a WSN testbed. Although we focused on fixed sensors positions, the on-line nature of the algorithm makes it suitable for use in dynamic environments where one seek to track the position of moving sensors.

{\appendix
\section{Proof of Proposition~\ref{prop:meanfield}}

Set for each $i$, $\sideset{}{_j}\sum \bs M(i,j)\bs U_{n-1}(j)=(\bs M\bs U_{n-1})_i$ and
\begin{align}
\label{eq:erreurs}
&\xi_n(i) = (\bs Y_n(i) -(\bs M \bs U_{n-1})_i )  + \bs U_{n-1}(i)(\bs U_{n-1}^T\bs M\bs U_{n-1} - \bs \Lambda_n(i)) 
\end{align}
Then, the sequence generated by each sensor node $i$ is written as:
\begin{align*}
\bs U_n(i) = \bs U_{n-1}(i) + \gamma_n \left((\bs M\bs U_{n-1})_i -  \bs U_{n-1}(i)(\bs U_{n-1}^T\bs M\bs U_{n-1}) \right)+\gamma_n\xi_n(i)
\end{align*}
Denote by $\cF_n$ the $\sigma$-algebra generated by all random variables defined up to time~$n$.
Using Lemmas~\ref{lem:sn}~\ref{lem:Munbiased} and~\ref{lem:YnSigman}, it is immediate to check
that $\bE(\xi_n|\cF_{n-1})=0$ and thus the sequence $\sum_{k\leq n} \gamma_k\xi_k$ is $\cF_n$-adapted martingale.
We estimate
\begin{align}
\label{eq:Exi}
\mathbb E[\|\xi_n(i)\|^2|\cF_{n-1}]\leq& \, \, \mathbb E[\|\bs Y_n(i)\|^2|\cF_{n-1}] + \|\bs U_{n-1}(i)\|^2\mathbb E[\|\bs \Lambda_n(i)\|^2|\cF_{n-1}] \nonumber\\
&\, \, +2\|\bs U_{n-1}(i)\|\mathbb E[\|\bs Y_n(i)\bs \Lambda_n(i)\| |\cF_{n-1}] \, .
\end{align}
The first term on the right hand side (RHS) of \eqref{eq:Exi} can be expanded as:
\begin{align}
\label{eq:Exi-1}
\mathbb E[\|\bs Y_n(i)\|^2|\cF_{n-1}]\leq &\,\, \mathbb E[|\bs{\widehat M}_n(i,i)|^2]\|\bs U_{n-1}(i)\|^2 \nonumber \\
 &\, \,+ \frac{N}{q}\sum_{j\neq i}\mathbb E[|\bs{\widehat M}_n(i,j)|^2]\|\bs U_{n-1}(i)\|^2 \nonumber \\
 &\, \,+2\sum_{j\neq i}\mathbb E[\bs{\widehat M}_n(i,i)\bs{\widehat M}_n(i,j)]\|\bs U_{n-1}(i)\|^2 \, .
\end{align}
Upon noting that for any $i,j$ $\mathbb E[\bs S_n(i,j)^2]=\frac{1}{q_{ij}}d_{i,j}^4C^8$ and $\bs U_{n-1}$ lies in a fixed compact set, there exists a constant $K'$ such that $\mathbb E[\|\bs Y_n(i)\|^2|\cF_{n-1}]\leq K'$ for all $n$ depending on $N$, $q_{min}=\min_{i,j} q_{ij}$, $C$ defined in \eqref{eq:d2nobias} and $\max_{i,j} d_{i,j}^4$ such that $\mathbb E[|\bs{\widehat M}_n(i,j)|^2]<K$ for some constant $K$. 
The second term on the RHS of \eqref{eq:Exi} can be handled similarly:
\begin{align}
\label{eq:Exi-2}
\mathbb E[\|\bs \Lambda_n(i)\|^2|\cF_{n-1}]& \leq \, \mathbb E[\|\bs Y_n(i)\|^2|\cF_{n-1}]\|\bs U_{n-1}(i)\|^2 +(\frac{N}{q}\sum_{j\neq i}\mathbb E[\|\bs Y_n(j)\|^2|\cF_{n-1}] \nonumber\\
& \quad \, \, +2\sum_{j\neq i}\mathbb E[\bs Y_n(i)\bs Y_n(j)|\cF_{n-1}] )\|\bs U_{n-1}(j)\|^2 \ \ \leq K''
\end{align}
for some constant $K''$.
Finally,
\begin{align}
\label{eq:Exi-3}
\mathbb E[\|\bs Y_n(i)\bs \Lambda_n(i)\| |\cF_{n-1}] \leq & \,  \, \mathbb E[\|\bs Y_n(i)\|^2|\cF_{n-1}]\|\cF_{n-1}(i)\| \nonumber \\
& \, \, +\sum_{j\neq i}\mathbb E[\bs Y_n(i)\bs Y_n(j)|\cF_{n-1}]\|\cF_{n-1}(j)\| 
\end{align}
is uniformly bounded as well. Therefore, we have shown that a.s.
$$
\sup _n  \mathbb E[\|\xi_n(i)\|^2|\cF_{n-1}]<\infty 
$$
Since $\sum_n \gamma_n^2<\infty$, it
follows that $\sum_n \gamma_n^2\mathbb E[\|\xi_n(i)\|^2|\cF_{n-1}]<\infty$ a.s. 
By Doob's Theorem, the martingale $\sum_{k\leq n}  \gamma_k \xi_k(i)$ converges almost surely
to some random variable finite almost everywhere. This completes the proof.

\section{Proof of Theorem~\ref{theo:convergence}}

Consider the following Lyapunov function $V:\mathbb R^{N\times p}\smallsetminus \{0\} \to \mathbb R^+$:
\begin{align}
\label{eq:lyp}
V(\bs U) = \frac{e^{\|\bs U\|^2}}{\bs U^T\bs M\bs U}\,.
\end{align}
The following properties hold:
\begin{enumerate}[label={\roman*)}]
\item $\lim_{\|\bs U\|\to \infty} V(\bs U) = +\infty$ and the gradient is $\nabla V(\bs U) = -2\frac{V(\bs U)}{\bs U^T\bs M\bs U}h(\bs U)$.
\item $\la V(\bs U),h(\bs U) \ra\leq 0$ and the equality holds iff $\{\bs U \in \mathbb R^{N\times p}\, |\,h(\bs U)=0\}$.
\end{enumerate}
The proof is an immediate consequence of Proposition~\ref{prop:meanfield}, the existence of~\eqref{eq:lyp} along with Theorem 2 of \cite{delyon:2000}. Sequence $\bs U_n$ converges a.s. to the roots of $h$. The latter roots are characterized in~\cite{oja:1992}. In particular, $h(\bs U)=0$ implies that each column of $\bs U$ is an unit-norm eigenvector of $\bs M$.
}
%

\section*{References}
  \bibliographystyle{elsarticle-num} 
  \bibliography{bib_longue}

\begin{thebibliography}{10}
\expandafter\ifx\csname url\endcsname\relax
  \def\url#1{\texttt{#1}}\fi
\expandafter\ifx\csname urlprefix\endcsname\relax\def\urlprefix{URL }\fi
\expandafter\ifx\csname href\endcsname\relax
  \def\href#1#2{#2} \def\path#1{#1}\fi

\bibitem{mds:borg}
I.~Borg, P.~Groenen, {Modern Multidimensional Scaling: theory and
  applications}, New York: Springer-Verlag, 1997.

\bibitem{shang:2003}
Y.~Shang, W.~Ruml, M.~Fromherz, Localization from mere connectivity, in:
  Proceedings of the 4th ACM International Symposium on Mobile Ad Hoc
  Networking \&Amp; Computing, MobiHoc '03, ACM, 2003, pp. 201--212.

\bibitem{rappaport:2002}
T.~Rappaport, {Wireless Communications: Principles and Practice}, Prentice
  Hall, 1996.

\bibitem{patwari:2005}
N.~Patwari, et~al., {Locating the Nodes : cooperative localization in wireless
  sensor networks}, IEEE Signal Processing Magazine 22~(4) (2005) 54--69.

\bibitem{mao:2007}
G.~Mao, B.~Fidan, B.~Anderson, {Wireless sensor network localization
  techniques}, Computer Networks 51~(10) (2007) 2529--2553.

\bibitem{multi:2011}
W.~Hereman, {Trilateration: The Mathematics Behind a Local Positioning System
  }, seminar (June 2011).

\bibitem{savvides:2002}
A.~Savvides, H.~Park, M.~Srivastava, {The Bits and Flops of the N-hop
  Multilateration Primitive For Node Localization Problems}, in: Proceedings of
  the 1st ACM International Workshop on Wireless Sensor Networks and
  Applications, WSNA '02, ACM, 2002, pp. 112--121.

\bibitem{niculescu:2001}
D.~Niculescu, B.~Nath, {Ad Hoc Positioning System (APS)}, in: IN GLOBECOM,
  2001, pp. 2926--2931.

\bibitem{savarese:2002}
C.~Savarese, J.~Rabaey, K.~Langendoen, {Robust positioning algorithm for
  distributed ad-hoc wireless sensor network}, in: Proceedings of the General
  Track of the Annual Conference on USENIX Annual Technical Conference,
  Monterey, 2002, pp. 317--327.

\bibitem{patwari:2001}
N.~Patwari, R.~J'Odea, W.~Yanwei, {Relative location in wireless networks}, in:
  VTC, 2001.

\bibitem{kamol:2004}
K.~Kaemarungsi, P.~Krishnamurthy, {Modeling of Indoor Positioning Systems Based
  on Location Fingerprinting}, in: INFOCOM, 2004.

\bibitem{mle:2010}
J.~Xu, W.~Liu, F.~Lang, Y.~Zhang, C.~Wang, Distance measurement model based on
  rssi in wsn., Wireless Sensor Network 2~(8) (2010) 606--611.

\bibitem{amy:2013}
N.~Dieng, M.~Charbit, C.~Chaudet, L.~Toutain, T.~Meriem, {Indoor Localization
  in Wireless Networks based on a Two-modes Gaussian Mixture Model}, in: IEEE
  78th Vehicular Technology Conference (VTC Fall), 2013, pp. 1--5.

\bibitem{leeuw:1977}
J.~Leeuw, {Applications of Convex Analysis to Multidimensional Scaling.},
  Recent Developments in Statistics (1977) 133--145.

\bibitem{biswas:2006:cent}
P.~Biswas, T.~Liang, T.~Wang, Y.~Ye, {Semidefinite programming based algorithms
  for sensor network localization}, ACM Transactions on Sensor Networks 2.

\bibitem{costa:2005}
J.~Costa, N.~Patwari, A.~Hero, {Distributed Weighted-Multidimensional Scaling
  for Node Localization in Sensor Networks}, ACM Transactions on Sensor
  Networks 2~(1) (2006) 39--64.

\bibitem{cedric:2007}
M.~Essoloh, C.~Richard, H.~Snoussi, Localisation distribuée dans les réseaux
  de capteurs sans fil par résolution d'un problème quadratique, in: GRETSI,
  2007.

\bibitem{biswas:2006:dist}
P.~Biswas, Y.~Ye, {A Distributed Method for Solving Semidefinite Programs
  Arising from Ad Hoc Wireless Sensor Network Localization}, in: Multiscale
  Optimization Methods and Applications, Vol.~82 of Nonconvex Optimization and
  Its Applications, Springer US, 2006, pp. 69--84.

\bibitem{sayed:2010}
F.~Cattivelli, A.~Sayed, {Distributed nonlinear Kalman filtering with
  applications to wireless localization}, in: Acoustics Speech and Signal
  Processing (ICASSP), 2010 IEEE International Conference on, Dallas, TX, 2010,
  pp. 3522 -- 3525.

\bibitem{gianackis:2009}
N.~Trawny, S.~Roumeliotis, G.~Giannakis, {Cooperative multi-robot localization
  under communication constraints}, in: Robotics and Automation, 2009. ICRA
  '09. IEEE International Conference on, Kobe, 2009, pp. 4394 -- 4400.

\bibitem{shang:2004}
Y.~Shang, W.~Ruml, {Improved MDS-based localization}, in: INFOCOM 2004.
  Twenty-third AnnualJoint Conference of the IEEE Computer and Communications
  Societes, Hong Kong, 2004, pp. 2640 -- 2651 vol.4.

\bibitem{montanari:2011}
S.~Korada, A.~Montanari, S.~Oh, Gossip pca, ACM SIGMETRICS Performance
  Evaluation Review 39~(1) (2011) 169--180.

\bibitem{cdc:2012}
G.~Morral, P.~Bianchi, J.~Jakubowicz, {Asynchronous Distributed Principal
  Component Analysis Using Stochastic Approximation}, in: Decision and Control
  (CDC), 2012 IEEE 51st Annual Conference on, Maui, Hawaii, 2012, pp. 1398 --
  1403.

\bibitem{oja:1992}
E.~Oja, {Principal components, minor components, and linear neural networks},
  Journal of Neural Networks 5~(6) (1992) 927--935.

\bibitem{borkar:meyn:2012}
V.~Borkar, S.~Meyn, Oja’s algorithm for graph clustering, markov spectral
  decomposition, and risk sensitive control, Journal of Automatica 48~(10)
  (2012) 2512--2519.

\bibitem{borkar:2008}
V.~Borkar, Stochastic approximation: a dynamical system viewpoint, Cambridge
  University Press, 2008.

\bibitem{tsitsiklis:phd-1984}
J.~Tsitsiklis, {Problems in Decentralized Decision Making and Computation},
  Ph.D. thesis, Massachusetts Institute of Technology (1984).

\bibitem{dsa:2013}
P.~Bianchi, G.~Fort, W.~Hachem, {Performance of a {D}istributed {S}tochastic
  {A}pproximation {A}lgorithm}, IEEE Transactions on Information Theory 59~(11)
  (2013) 7405 -- 7418.

\bibitem{dsa:2014}
G.~Morral, P.~Bianchi, G.~Fort, {Success and Failure of Adaptation-Diffusion
  Algorithms for Consensus in Multi-Agent Networks}, Arxiv preprint
  arXiv:1410.6956.

\bibitem{nedic:olshevsky:2014}
A.~Nedic, A.~Olshevsky, {Stochastic Gradient-Push for Strongly Convex Functions
  on Time-Varying Directed Graphs}, Arxiv preprint arXiv:1406.2075.

\bibitem{delyon:2000}
B.~Delyon, {Stochastic Approximation with Decreasing Gain: Convergence and
  Asymptotic Theory}, Unpublished Lecture Notes,
  http://perso.univ-rennes1.fr/bernard.delyon/as\_cours.ps.

\end{thebibliography}

\newpage
\renewcommand\thefigure{\arabic{figure}}
\setcounter{figure}{0}

\renewcommand\thetable{\arabic{table}}
\setcounter{table}{0}

\begin{figure}[h!]
\centering
 \includegraphics[width=0.8\textwidth,height=5cm]{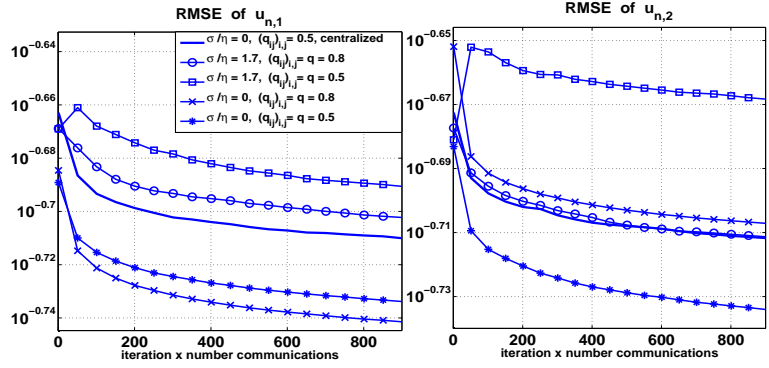}
 \caption{RMSE as a function of $nN$ from the two estimated eigenvectors $\bs u_{n,1}$ and $\bs u_{n,2}$ when considering the noiseless and noisy case and for different values of~$q$. }
 \label{fig:comparison_veps}
\end{figure}

\begin{figure}[h!]
        \centering
                \begin{subfigure}[b]{\textwidth}
                        \centering
\includegraphics[width=0.65\textwidth,height=5cm]{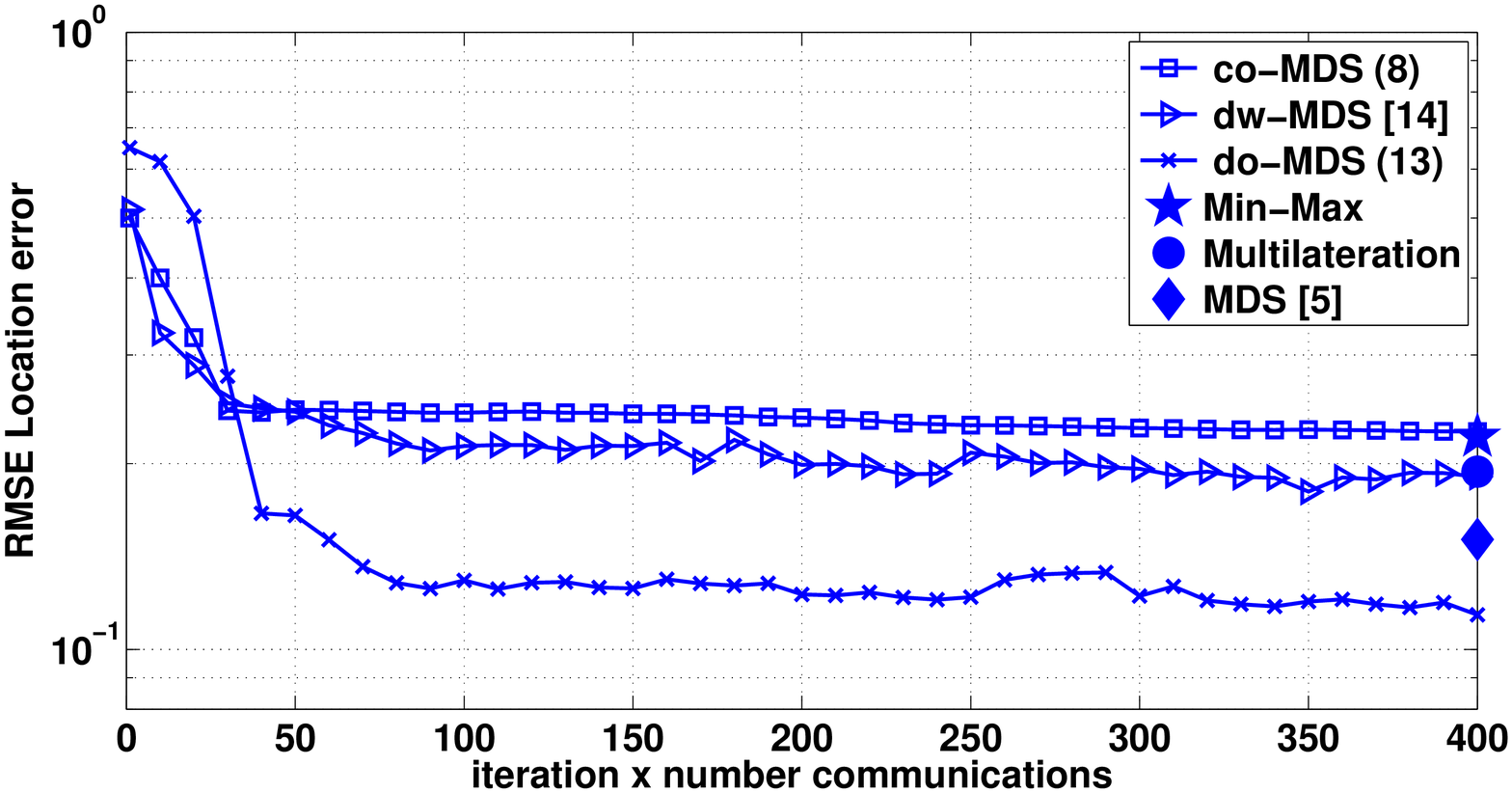}    
\caption{RMSE as a function of $nN$ from the estimated positions $(\bs{\widehat Z}_n(1),\dots, \bs{\widehat Z}_n(N))$. }
        \label{subfig:comparison}
        \end{subfigure} \\
        \begin{subfigure}[b]{0.31\textwidth}
                 \includegraphics[width=\textwidth,height=4.5cm]{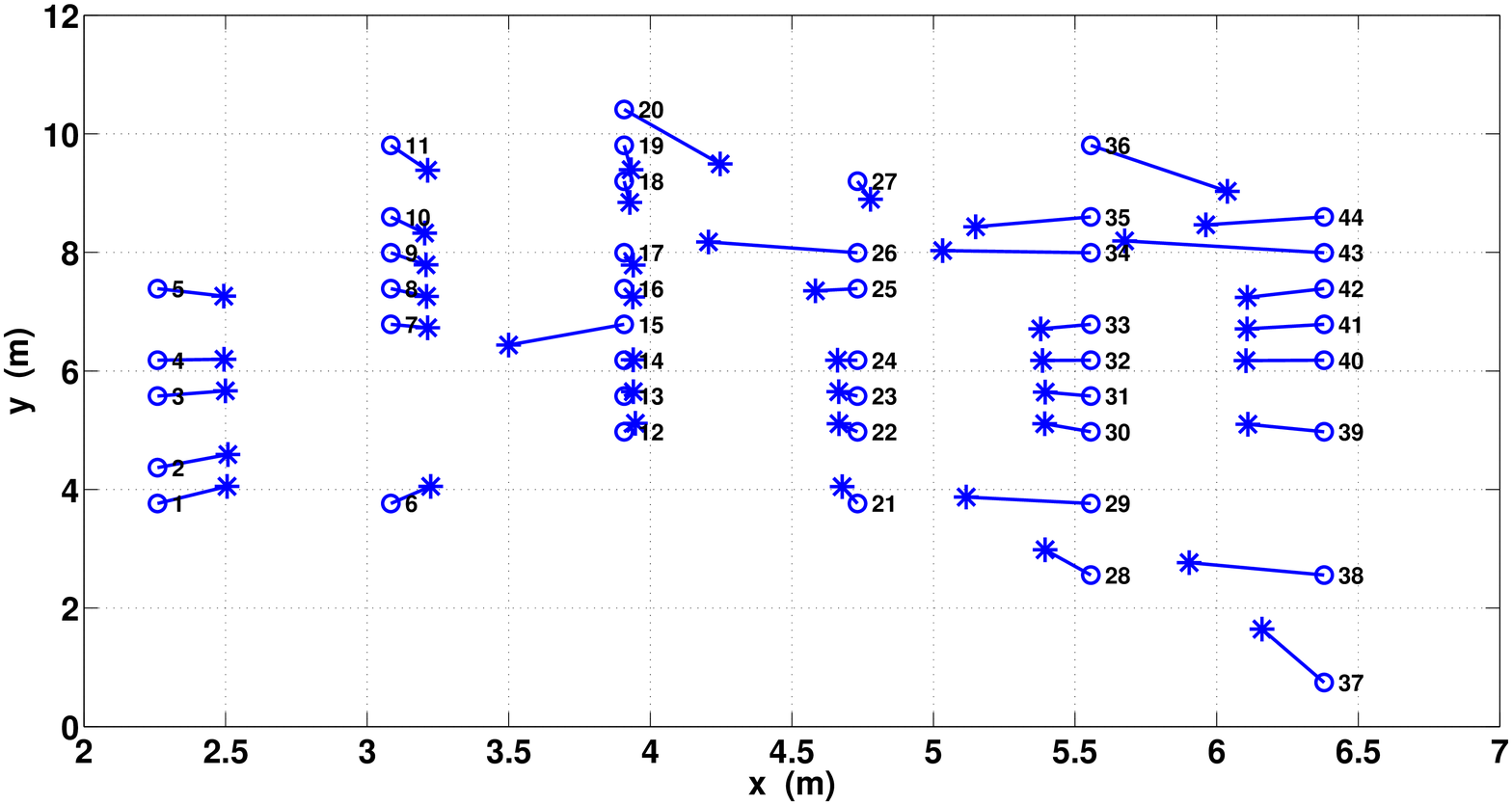}
     \caption{Oja's algorithm~\eqref{eq:oja}-\eqref{eq:oja2}.}
        \label{subfig:sim_comds}
        \end{subfigure}%
        ~ 
        \begin{subfigure}[b]{0.31\textwidth}
                \includegraphics[width=\textwidth,height=4.5cm]{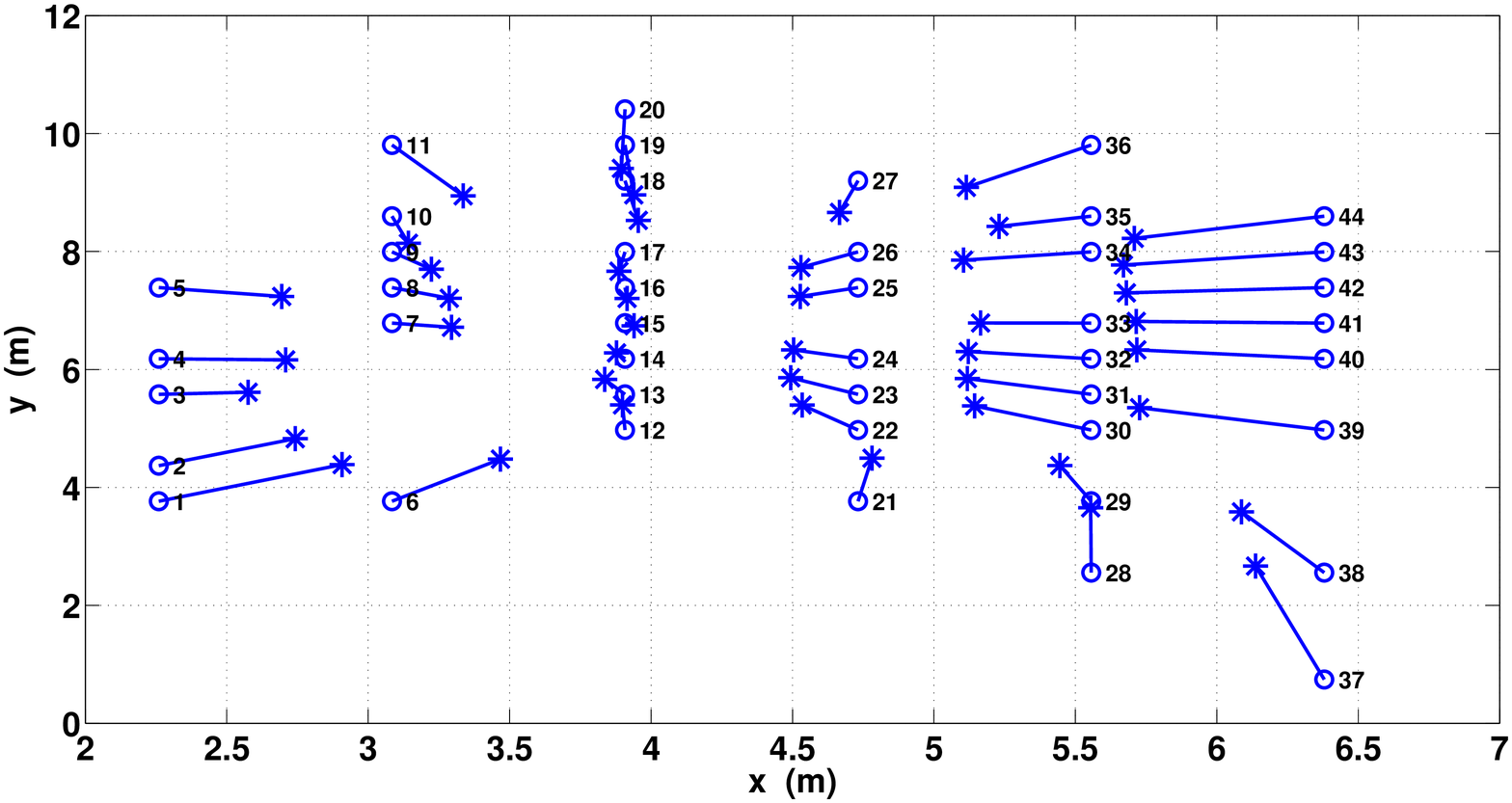}
     \caption{dwMDS \cite{costa:2005}.}
        \label{subfig:sim_dwdms}
        \end{subfigure}
         ~ 
        \begin{subfigure}[b]{0.31\textwidth}
                \includegraphics[width=\textwidth,height=4.5cm]{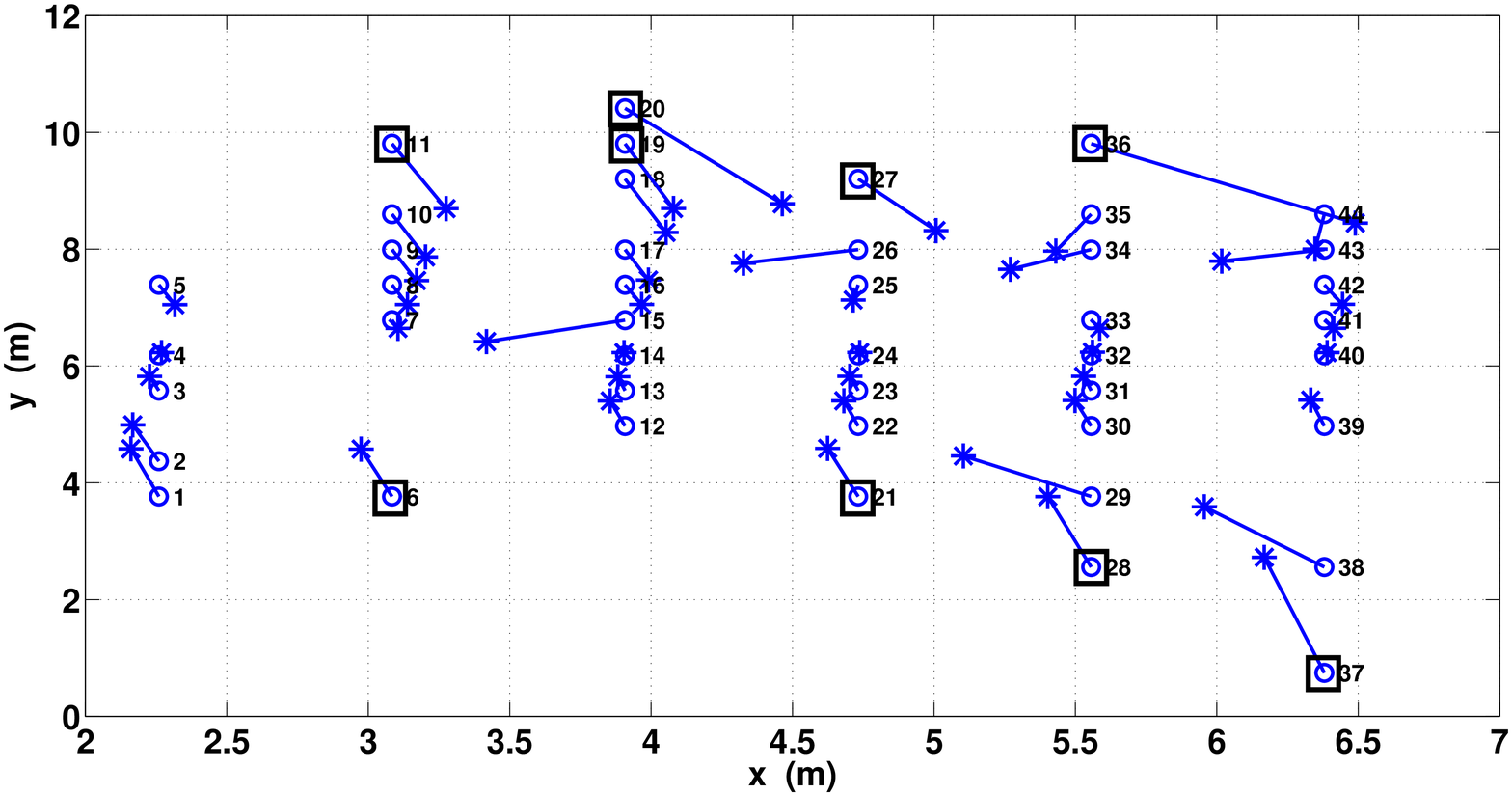}
     \caption{Algorithm~\ref{alg:dmds} (doMDS).}
        \label{subfig:sim_domds}
        \end{subfigure}
       \caption{Estimated positions after $1000$ iterations. Markers ({\color{blue}\ding{81}}) correspond to the estimated values while markers ({\color{blue} $\bs{\ocircle}$}) to the true positions. Squared positions ($\square$) in d) highlight worse accuracy compared to b).}
       \label{fig:sim_networks}
\end{figure}

\begin{table}[h!]
        \centering
         \begin{tabular}{ | l | c | c | c | c | c |c| }
    \hline
      \textbf{Method} & \textbf{MC}& \textbf{min-max}&\textbf{MDS}&\textbf{Oja}& \textbf{dwMDS} & \textbf{doMDS} \\ \hline
      \textbf{Before refinement}  & $1.87$ & $0.8$& $1.98$ & $2.18$& $0.86$& $1.56$\\ \hline
      \textbf{After refinement}  & $1.05$ & $0.54$& $1.39$ & $1.37$& $0.6$& $0.51$ \\ \hline
        \textbf{Improvement (\%) }  & $44$ & $32$& $30$ & $28$& $30$& $78$ \\ \hline
         \textbf{Positions improved (\%) }  & $75$ & $71$& $80$ & $80$& $82$& $86$ \\ \hline
     \end{tabular}
     \caption{RMSE averaged over the $44$ estimated positions considering real data.}
\label{tab:comp_ref}
\end{table}
 
\end{document}